\documentclass[a4paper,11pt]{article}
\usepackage[utf8]{inputenc}

\usepackage{amsmath}
\usepackage{amsthm}
\usepackage{amssymb}

\newtheorem{lemma}{Lemma}
\newtheorem{theorem}{Theorem}
\newtheorem{corollary}{Corollary}
\newtheorem{definition}{Definition}

\title{Graph Reconstruction via MIS Queries}
\author{Christian Konrad\footnote{School of Computer Science, University of Bristol, \texttt{christian.konrad@bristol.ac.uk}} \and Conor O'Sullivan\footnote{School of Computer Science, University of Bristol, \texttt{je20344@bristol.ac.uk}} \and Victor Traistaru\footnote{School of Computer Science, University of Bristol, \texttt{victoryt2001@gmail.com}}}

\date{}

\DeclareMathOperator{\poly}{poly}

\usepackage{algorithm}
\usepackage{algorithmic}
\usepackage{fullpage}

\begin{document}

\maketitle

\begin{abstract}
In the \textsf{Graph Reconstruction} (\textsf{GR}) problem, a player initially only knows the vertex set $V$ of an input graph $G=(V, E)$ and is required to learn its set of edges $E$. To this end, the player submits queries to an oracle and must deduce $E$ from the oracle's answers.

Angluin and Chen [Journal of Computer and System Sciences, 2008] resolved the number of \textsf{Independent Set} (\textsf{IS}) queries necessary and sufficient for \textsf{GR} on $m$-edge graphs. In this setting, each query consists of a subset of vertices $U \subseteq V$, and the oracle responds with a boolean, indicating whether $U$ is an independent set in $G$. They gave algorithms that use $O(m \cdot \log n)$ \textsf{IS} queries, which is best possible.

In this paper, we initiate the study of \textsf{GR} via \textsf{Maximal Independent Set} (\textsf{MIS}) queries, a more powerful variant of \textsf{IS} queries. Given a query $U \subseteq V$, the oracle responds with any, potentially adversarially chosen, maximal independent set $I \subseteq U$ in the induced subgraph $G[U]$.

We show that, for \textsf{GR}, \textsf{MIS} queries are strictly more powerful than \textsf{IS} queries when parametrized by the maximum degree $\Delta$ of the input graph. We give tight (up to poly-logarithmic factors) upper and lower bounds for this problem:
\begin{enumerate}
 \item We observe that the simple strategy of taking uniform independent random samples of $V$ and submitting those to the oracle yields a non-adaptive randomized algorithm that executes $O(\Delta^2 \cdot \log n)$ queries and succeeds with high probability. This should be contrasted with the fact that $\Omega(\Delta \cdot n \cdot \log(\frac{n}{\Delta}))$ \textsf{IS} queries are required for such graphs, which shows that \textsf{MIS} queries are strictly more powerful than \textsf{IS} queries. Interestingly, combining the strategy of taking uniform random samples of $V$ with the probabilistic method, we show the existence of a deterministic non-adaptive algorithm that executes $O(\Delta^3 \cdot \log(\frac{n}{\Delta}))$ queries.
 
 \item Regarding lower bounds, we prove that the additional $\Delta$ factor when going from randomized non-adaptive algorithms to deterministic non-adaptive algorithms is necessary. We show that every non-adaptive deterministic algorithm requires $\Omega(\Delta^3 / \log^2 \Delta)$ queries.
 For arbitrary randomized adaptive algorithms, we show that $\Omega(\Delta^2)$ queries are necessary in graphs of maximum degree $\Delta$, and that $\Omega(\log n)$ queries are necessary, even when the input graph is an $n$-vertex cycle. 
\end{enumerate}
 \end{abstract}
 
 \newpage

\section{Introduction}
Query algorithms for graph problems have recently received significant attention. In this setting,  algorithms are granted access to the input graph solely via a (usually easy-to-compute) subroutine, which is referred to as the oracle, and the complexity of an algorithm is measured by the number of subroutine/oracle calls.  

In the literature, a large number of different query models have been considered. Queries can either be {\em local} or {\em global}, depending on whether they reveal only local information, e.g., vertex degree queries \cite{f06} or queries that, for any $i$, reveal the $i$th neighbour of a vertex \cite{gr06}, or global information, e.g., (bipartite) independent set queries \cite{ac08,bhrrs18,ab19, amm22} or maximal matching queries \cite{bk20, kns23}. Depending on the application, it is generally desirable to obtain {\em non-adaptive} query algorithms, i.e., algorithms where the different queries do not depend on each other's outcomes, since such queries can be executed simultaneously and therefore admit straightforward parallel implementations. We also say that a query algorithm requires {\em $k$ rounds of adaptivity}, for some integer $k$, if the queries executed by the algorithm can be partitioned into $k$ groups such that the queries in the $i$th group only depend on the outcomes of the queries in groups $1, \dots, i-1$. 

\subparagraph*{Graph Reconstruction.}
In this work, we consider query algorithms for the \textsf{Graph Reconstruction} (\textsf{GR}) problem. In \textsf{GR}, a player initially only knows the vertex set $V$ of a graph $G=(V, E)$ and is tasked with learning the edge set $E$, by submitting a sequence of adaptive or non-adaptive queries to the oracle, and by deducing the edge set $E$ from the oracle's answers. %The sequence of queries can be deterministic or randomized.

Query algorithms for \textsf{GR} have been extensively studied under various query models, including \textsf{Independent Set} (\textsf{IS}) queries \cite{gk97,bakaf01,abkrs04,aa05,ac08,ab19}, 
\textsf{Distance} queries \cite{rs07,kmz18,mz21,rlyw21}, 
 \textsf{Betweenness} queries \cite{abrs16,rylw22}, and in the setting where an algorithm receives random vertex-induced subgraphs or submatrices of the adjacency matrix \cite{ms22}. Very relevant to our work are \textsf{IS} queries, where a query consists of a subset of vertices $U \subseteq V$, and the oracle returns a boolean, indicating whether $U$ is an independent set in $G$. \textsf{IS} queries for \textsf{GR} were originally studied for recovering simple graphs, such as matchings and Hamiltonian cycles \cite{gk97,bakaf01,abkrs04}, since these settings have direct applications to genome sequencing. Angluin and Chen \cite{ac08} were the first to consider general graphs and showed that $m$-edge graphs can be reconstructed with $O(m \cdot \log n)$ \textsf{IS} queries via an adaptive deterministic algorithm or a randomized algorithm with limited adaptivity (see also \cite{ab19} who reduced the number of adaptive rounds), and this is also best possible.
 
 \subparagraph*{MIS Queries.} In this work, we inititate the study of \textsf{GR} under \textsf{Maximal Independent Set}  (\textsf{MIS}) queries. In this setting, similar to \textsf{IS} queries, a query also consists of a subset of vertices $U \subseteq V$. The oracle, however, responds with any, potentially adversarially chosen, maximal independent set $I \subseteq U$ in the subgraph induced by the queried vertices $G[U]$. 
 
 \textsf{MIS} queries are global queries and are similar in spirit to the \textsf{Maximal Matching} queries considered in \cite{bk20,kns23}, where the oracle returns a maximal matching in the subgraph induced by the query vertices \cite{bk20}, or in the subgraph spanned by the queried edges that are also contained in the input graph \cite{kns23}.
They are at least as powerful as \textsf{IS} queries since an \textsf{IS} query $U \subseteq V$ can be answered by an \textsf{MIS} oracle by exploiting the connection:
$$\text{IS-Query}(U) = \texttt{true} \quad \Longleftrightarrow \quad  \text{MIS-Query}(U) = U \ . $$
We observe that computing an \textsf{MIS} is a simple task that can be solved in linear time by a \textsf{Greedy} algorithm and even in sublinear time on graphs of bounded-neighborhood independence \cite{as19}. The \textsf{MIS} problem can also be solved efficiently in many restricted computational models, e.g., an \textsf{MIS} can be computed in $O(\log \log \Delta)$ rounds in the \textsf{Congested-Clique} model \cite{ggkmr18}, in $O(\log \Delta)$ rounds in the \textsf{LOCAL} model of distributed computing \cite{g16,rg20}, and in $O(\log \log n)$ passes in the semi-streaming model \cite{acgmw15,akns24}. This task can therefore be regarded as a building block of more elaborate algorithms and may be outsourced to an oracle that maintains an efficient implementation. 
 
Our key motivation is to answer the question as to whether \textsf{MIS} queries are strictly more powerful than \text{IS} queries, i.e., whether enhancing the oracle answer by providing a maximal independent set rather than just indicating whether the query set is independent yields significant savings in the number of queries required.

%\subsection{Motivation}
%Our motivation for this work stems from two seperate objectives:

%First, recent works by Binti-Khalil and Konrad \cite{} and Konrad et al. \cite{} considered \textsf{Maximal Matching} (\textsf{MM}) queries, where an algorithm submits either a subset of vertices or a subset of potential input edges to the oracle, and the oracle returns any maximal matching in the subgraph either induced by the queries vertices or spanned by the queried edges that are also contained in the input graph. Their goal is to use \textsf{MM} queries to compute large matchings, or, phrased differently, to obtain approximation algorithms for \textsf{Maximum Matching}. The \textsf{MM} query model is motivated by the fact that, in many computational models, it is computationally inexpensive to compute maximal matchings, but significantly more expensive to compute large matchings. 

%We are interested whether a similar 

%We are interested in the question as to how much information a maximal independent set reveals about the input graph, and the \textsf{GR} problem with \textsf{MIS} queries allows us to study this question.

%In the \textsf{GR} problem, an algorithm is required to learn all the edges of the input graph. When considering the class of $n$-vertex graphs of maximum degree $\Delta$, a query algorithm is therefore required to learn up to $\Theta(n \Delta)$ edges. 

\subsection{Our Results}
In this paper, we show that \textsf{MIS} queries can be strictly more powerful than \textsf{IS} queries, but this is also not always the case. 

As previously mentioned, Angluin and Chen showed that $O(m \log n)$ \textsf{IS} queries are sufficient for reconstructing a graph on $m$ edges. One of our lower bound results (\textbf{Theorem~\ref{thm:lb-1}}) implies that there are $m$-edge graphs that require $\Omega(m)$ \textsf{MIS} queries to be reconstructed, which show that, up to a logarithmic factor, \textsf{MIS} queries are not more powerful than \textsf{IS} queries on the class of $m$-edge graphs. However, when considering the class of graphs with maximum degree $\Delta$, for some integer $\Delta$, significant savings can be achieved:

%In this paper, we answer this question in the affirmative. We resolve the \textsf{MIS}-query complexity of \textsf{GR} up to poly-logarithmic factors, both for the class of adaptive randomized query algorithms and for the class of non-adaptive deterministic query algorithms. 

\begin{table}
\begin{center} \begin{tabular}{c|l|l}
 & Algorithm  & Lower Bound \\
 \hline
 Randomized & $O(\Delta^2 \log n)$ (\textbf{Theorem~\ref{thm:ub}}) & $\Omega(\Delta^2 + \log n)$ (\textbf{Theorems~\ref{thm:lb-2} and \ref{thm:lb-1}}) \\
 Deterministic & $O(\Delta^3 \log(\frac{n}{\Delta}))$ (\textbf{Corollary~\ref{cor:ub-2}}) & $\Omega(\Delta^3 / \log^2 \Delta)$ (\textbf{Corollary~\ref{cor:det-lb}})
\end{tabular}

\vspace{0.3cm}
\caption{Overview of our results. Our results are parametrized by the maximum degree $\Delta$ of the input graph. Both our algorithms are non-adaptive and require knowledge of $\Delta$ in advance. We also give adaptive counterparts with similar query complexity that do not require advanced knowledge of $\Delta$. The lower bounds for randomized algorithms hold for adaptive algorithms. The lower bound for deterministic algorithms holds for non-adaptive algorithms. \label{tab:results} }
\end{center} \end{table}

\vspace{0.1cm}
\noindent \textbf{Algorithms.} We give a randomized algorithm for reconstructing an $n$-vertex graph with maximum degree $\Delta$ that uses $O(\Delta^2 \log n)$ non-adaptive queries and succeeds with high probability\footnote{As it is standard, we say that an event related to the input graph $G=(V, E)$ occurs {\em with high probability} if the probability is $1-\frac{1}{\poly n}$, where $n = |V|$. } (\textbf{Theorem~\ref{thm:ub}}). This result shows that, for the class of graphs of maximum degree $\Delta$, \textsf{MIS} queries are stronger than \textsf{IS} queries since an information-theoretic lower bound similar to the one given in \cite{ac08} shows that $\Omega(\Delta n \log(\frac{n}{\Delta}))$ \textsf{IS} queries are needed to reconstruct a graph with maximum degree $\Delta$ (\textbf{Corollary~\ref{cor:is-queries}} in Appendix~\ref{app:is-queries}).  Observe that, for the class of constant degree graphs, $O(\log n)$ \textsf{MIS} queries are sufficient, but $\Omega(n \log n)$ \textsf{IS} queries are necessary. 

We also investigate whether randomization is necessary to obtain algorithms with low query complexity. Using the probabilistic method, we show that there exists a non-adaptive deterministic query algorithm that executes $O(\Delta^3 \log(\frac{n}{\Delta}))$ queries (\textbf{Corollary~\ref{cor:ub-2}}). 

Our non-adaptive randomized and deterministic algorithms assume that the maximum degree $\Delta$ is known in advance. This cannot be avoided since, as proved in Theorem~\ref{thm:lb-1}, any algorithm that executes $o(n^2)$ queries cannot solve all instances with $\Delta = \Theta(n)$. We show that the non-adaptive algorithms above can be turned into adaptive algorithms with query complexities that are by at most a constant factor larger. These adaptive algorithms do not require any information about $\Delta$ in order to operate and only require $O(\log \Delta)$ rounds of adaptivity (\textbf{Corollary~\ref{cor:adaptive-algorithms}}).

\vspace{0.1cm}
\noindent \textbf{Lower Bounds.} We give lower bounds on the number of queries required by both non-adaptive deterministic  algorithms and adaptive randomized algorithms. 

First, we show that every non-adaptive deterministic query algorithm requires $\Omega(\Delta^3 / \log^2 \Delta)$ queries (\textbf{Corollary~\ref{cor:det-lb}}), which renders our deterministic algorithm optimal, up to poly-logarithmic factors. This result together with our non-adaptive randomized algorithm (\textbf{Theorem~\ref{thm:ub}}) establishes a separation result between non-adaptive randomized and non-adaptive deterministic algorithms since our randomized algorithm only requires $O(\Delta^2 \log n)$ queries.

Next, we show that every adaptive randomized query algorithm requires $\Omega(\log n)$ queries, even if the input graph is guaranteed to be an $n$-vertex cycle (\textbf{Theorem~\ref{thm:lb-2}}). 
Furthermore, we show that every adaptive randomized query algorithm requires $\Omega(\Delta^2)$ queries, for any $\Delta$ (\textbf{Theorem~\ref{thm:lb-1}}).
These lower bounds show that the number of queries executed by our randomized algorithm is optimal, up to a logarithmic factor, and that the logarithmic dependency on $n$ cannot be avoided entirely. %We leave it as an open question as to whether a randomized algorithm that executes only $O(\Delta^2 + \log n)$ queries exists.

For an overview of our results, see Table~\ref{tab:results}.

\subsection{Techniques}
\textbf{Algorithms.} The starting point for all our algorithms is an algorithm by Angluin and Chen \cite{ac08} for \textsf{GR} that uses $\textsf{IS}$ queries. Our randomized algorithm is in fact identical to their algorithm, up to the use of a different sampling probability. It works as follows: %Our deterministic algorithm The algorithm also serves as the starting point for our deterministic algorithm as well as our adaptive algorithms that do not require advanced knowledge of $\Delta$.

The key idea is to learn all the non-edges rather than all the edges of the input graph. To this end, we  sample sufficiently many random vertex-induced subgraphs $G[V_i] \subseteq G$, for integers $i$, such that, for every non-edge $uv \in (V \times V) \setminus E(G)$, the probability that both $u$ and $v$ are contained in $V_i$ and are isolated in $G[V_i]$ is $\Theta(\frac{1}{\Delta^2})$. This is achieved by including every vertex in $V_i$ with probability $\frac{1}{\Delta+1}$. Observe that when $u$ and $v$ are isolated in $G[V_i]$ then the oracle necessarily needs to include both $u$ and $v$ in the returned maximal independent set. The returned maximal independent set therefore serves as a witness that proves that the potential edge $uv$ is not contained in the input graph. We then argue that, after repeating this process $O(\Delta^2 \log n)$ times, we have learnt all non-edges with high probability, which allows us to identify the edges of the input graph by complementing the set of non-edges. 

Angluin and Chen use this idea in the context of \textsf{IS} queries for graphs with maximum degree $\Delta$, where the significantly lower inclusion probability of each vertex into each sample $V_i$  of $\frac{1}{\sqrt{\Delta n}}$ needs to be used. 

Our deterministic algorithm is built on a similar idea. We show that, when taking $\ell = \Theta(\Delta^3 \log(\frac{n}{\Delta}))$ random subsets $V_1, \dots, V_{\ell} \subseteq V$ as above, i.e., by inserting every vertex $v \in V$ into any subset $V_i$ with probability $\frac{1}{\Delta+1}$, then, for any tuple $(\{u, v\}, \{w_1, \dots, w_{2\Delta}\})$ of $2\Delta+2$ disjoint vertices, which we also refer to as a {\em witness}, there exists a set $V_i$ such that $u,v \in V_i$ but $w_1, \dots, w_{2\Delta} \notin V_i$ with positive probability. Since this event happens with non-zero probability, such a family of subsets $V_1, \dots, V_{\ell}$ exists. Our deterministic algorithm then queries all of these subsets $(V_i)_{1 \le i \le \ell}$. We now claim that our algorithm learns every non-edge $uv \in (V \times V) \setminus E$: We know that there exists a set $V_i$ such that $u,v \in V_i$ but $(\Gamma(v) \cup \Gamma(u)) \cap V_i = \varnothing$ since $|\Gamma(v) \cup \Gamma(u)| \le 2\Delta$ holds. The vertices $u$ and $v$ are thus isolated in $G[V_i]$ and therefore necessarily reported in the oracle answer, which provides a proof to the algorithm that the edge $uv$ is not contained in the input graph.

Both our non-adaptive randomized and deterministic algorithms require advanced knowledge of $\Delta$, which, as previously argued, is unavoidable. We turn both of these algorithms into adaptive algorithms that require only $O(\log \Delta)$ rounds of adaptivity and whose total number of queries executed is only by a constant factor larger. This is achieved by successively running our algorithms for the guesses $D = 1, 2, 4, 8, \dots$ for the maximum degree $\Delta$ until a final guess $\Delta < D \le 2 \Delta$ is used, in which case it is easy to see that the graph is correctly reconstructed. While this is a standard doubling argument, the non-trivial part of the argument is to identify when the condition $\Delta < D \le 2 \Delta$ is reached since $\Delta$ is now unknown. Denoting by $F$ the set of non-edges learnt by the execution of one of our algorithms when invoked on the current guess $D$, we show that if the maximum degree in the graph spanned by the edges $(V \times V) \setminus F$ is at least $D$ then we have indeed learnt all the non-edge of the input graph, and thus also reconstructed the graph.

\vspace{0.1cm}
\noindent \textbf{Lower Bounds.} 
We first discuss our $\Omega(\Delta^3 / \log^2 \Delta)$ lower bound for deterministic non-adaptive algorithms. Witnesses, i.e., tuples $(\{u, v\}, \{w_1, \dots, w_{2\Delta}\})$ of $2\Delta+2$ disjoint vertices, play a key role in our lower bound argument. We argue that any deterministic non-adaptive query algorithm must be such that, for every witness $(\{u, v\}, \{w_1, \dots, w_{2\Delta}\})$ of disjoint $2\Delta+2$ vertices, there exists a query $Q_i$ such that $u,v \in Q_i$ and $\{w_1, \dots, w_{2\Delta} \} \cap Q_i = \varnothing$. We call a set of $\ell$ queries that fulfills this property a {\em $\Delta$-Query-Scheme of size $\ell$}. To see that the previous property is true, for the sake of a contradiction, suppose that this is not the case and there exists a witness $(\{u, v\}, \{w_1, \dots, w_{2\Delta}\})$ that does not fulfill these properties, i.e., whenever $u,v$ is included in a query then at least one of the vertices $w_1, \dots, w_{2\Delta}$ is included in this query as well. We claim that the two graphs $G_1$ and $G_2$, where, in both graphs, $w_1, \dots, w_{\Delta}$ are incident on $u$ and $w_{\Delta+1}, \dots, w_{2\Delta}$ are incident on $v$, and $uv$ is an edge in $G_1$ but not in $G_2$ cannot be distinguished\footnote{This construction generates a maximum degree of $\Delta+1$ in $G_1$, which is technically not allowed since we consider algorithms that run on graphs of maximum degree $\Delta$. To circumvent this issue, in our actual proof we therefore relate algorithms that operate on maximum degree $\Delta$ graphs to $(\Delta-1)$-Query-Schemes.}. Indeed, we claim that, for any query executed, the oracle can always report an independent set that contains at most one of the two vertices $u,v$, even on graph $G_2$. This is because whenever $u,v$ is included in a query then at least one of the vertices $(w_i)_{1 \le i \le 2\Delta}$ is also included in this query. The oracle can therefore include $w_i$ in the oracle answer, which implies that at most one of $u$ and $v$ will also be included. Both graphs $G_1$ and $G_2$ are therefore consistent with all oracle answers and are thus indistinguishable, a contradiction.

Our task, therefore, is to prove that any $\Delta$-Query-Scheme must be of size at least $\Omega(\Delta^3 / \log^2 \Delta)$. We achieve this by combining two separate arguments that address small queries, i.e., queries $V_i \subseteq V$ of size at most $|V_i| \le t := C \cdot \frac{n \ln(\Delta)}{\Delta}$ and large queries of size larger than $t$, respectively: For any pair of vertices $\{u,v\}$, small queries are such that they cover many witnesses $(\{u,v\}, \{w_1, \dots, w_{2\Delta} \})$, for many different vertices $(w_i)_{1 \le i \le 2\Delta}$, since the vertices $(w_i)_{1 \le i \le 2 \Delta}$ are chosen from $V \setminus V_i$, however, on average, vertex pairs $\{u,v\}$ cannot be included in many small queries (at most ${t \choose 2}$). In contrast, vertex pairs $\{u,v\}$ can be included in many large queries, however, those queries do not cover many witnesses $(\{u,v\}, \{w_1, \dots, w_{2\Delta} \})$, for vertices $(w_i)_{1 \le i \le 2\Delta}$. We obtain our result by combining these observations, see the proof of Lemma~\ref{lem:lb-query-scheme}.

Next, we discuss our $\Omega(\log n)$ queries lower bound for adaptive randomized algorithms on $n$-vertex cycles. Observe that, in an $n$-vertex cycle $G=(V, E)$, every pair of vertices $u,v \in V$ is such that $\Gamma(u) \neq \Gamma(v)$. Hence, for a query algorithm to be able to distinguish any two vertices of the input graph, the algorithm is required to obtain oracle answer maximal independent sets such that $u$ and $v$ do not behave the same in all returned independent sets. We show that $\Omega(\log n)$ queries are needed to distinguish between all the vertices. Our argument is based on the observation that, given a set of vertices $U \subseteq V$ that are so far indistinguishable and a query set $V_i \subseteq V$, the returned maximal independent set $I_i$ only allows us to differentiate between the vertices of $U$ that are not queried, queried and contained in $I_i$, and queried and not contained in $I_i$. Every query thus allows us to partition any set of vertices that is still indistinguishable into only three subsets, which implies that $\Omega(\log_3 n) = \Omega(\log n)$ queries are needed to distinguish between all the vertices.

Last, we discuss our $\Omega(\Delta^2)$ queries lower bound for arbitrary randomized  adaptive algorithms. We work with the family of graphs $\mathcal{G}$ obtained from the family of all balanced bipartite graphs $G= (A, B, E)$ with $|A| = |B| = \Theta(\Delta)$, where the two bipartitions $A$ and $B$ are turned into two separate cliques. Observe that each such graph has a maximum independent set size of $2$. 
We now claim that, for every non-edge $ab \in (V \times V) \setminus E(G)$ in the input graph $G$, there must exist an independent set $I_i$ returned by the oracle such that $I_i = \{a, b\}$. This needs to be the case since otherwise the algorithm cannot distinguish between $G$ and the graph $G' = G \cup \{ab\}$, i.e., $G$ with the edge $ab$ added, since all query responses are then consistent with both $G$ and $G'$. However, since at most one non-edge can be learnt per query, and $\mathcal{G}$ contains many graphs with $\Omega(\Delta^2)$ non-edges, the result follows. Our actual lower bound is proved for randomized algorithms via an application of Yao's Lemma, which makes the argument slightly more complicated.

\subsection{Recent Developments}
Since a preprint of our work was released on arXiv in early 2024, Michael and Scott \cite{ms24} improved some of our lower bounds by poly-logarithmic factors. They show that, for arbitrary adaptive randomized algorithms, $\Omega(\Delta^2 \log(n/\Delta) / \log \Delta)$ queries are needed, improving our $\Omega(\Delta^2)$ lower bound by a $\log(n/\Delta) / \log \Delta$ factor. They further strengthen this bound to $\Omega(\Delta^2 \log(n / \Delta))$ for non-adaptive randomized algorithms. For deterministic non-adaptive algorithms, they give a $\Omega(\Delta^3 \log n / \log \Delta)$ lower bound, improving over our $\Omega(\Delta^3 / \log^2 \Delta)$ lower bound.

Their improvements for randomized algorithms are obtained by using different and more involved graph constructions. Their lower bound for non-adaptive deterministic algorithms is obtained by establishing a connection between what we refer to as $\Delta$-Query-Schemes (which are equivalent to non-adaptive deterministic algorithms) and cover-free families, which are well-studied mathematical objects. They give an improved lower bound for a cover-free family with certain parameters, which, by the novel connection identified, translates to a lower bound for deterministic non-adaptive algorithms.

%\subsection{Other Query Algorithms for Graph Reconstruction}
%Query algorithms for the \textsf{GR} problem have been extensively studied under distance queries, where pairs of vertices in a connected input graph are queried and the oracle returns their distance  \cite{rs07,kmz18,mz21,rlyw21}. In this setting, it is known that general graphs cannot be reconstructed with $o(n^2)$ queries \cite{rs07}, but bounded-degree graphs can be reconstructed with $O(n^{1.5})$ queries \cite{kmz18} and chordal graphs with $O(\Delta^2 n \log n)$ queries \cite{rlyw21}, where $\Delta$ denotes the maximum degree of the input graph. 

%Angluin and Chen consider the \textsf{GR} problem using an \textsf{Independent Set} oracle, where a subset of vertices is submitted to the oracle and the oracle responds with a boolean indicating whether the set is an independent set. They show that, a graph with $m$ edges can be reconstructed using $O(m \log n)$ queries, which is optimal (see also \cite{ab19} for a follow-up work).

%The \textsf{GR} problem has also very recently been studied when the algorithm is provided with either random vertex-induced subgraphs or submatrices of the adjacency matrix, see \cite{ms22}.

\subsection{Outline}
We give our algorithms in Section~\ref{sec:ub} and our lower bounds in Section~\ref{sec:lb}. We conclude with open problems in Section~\ref{sec:conclusion}.

\section{Algorithms}\label{sec:ub}
We give our randomized non-adaptive algorithm that executes $O(\Delta^2 \log n)$ queries in Subsection~\ref{sec:alg-rand}, and our deterministic non-adaptive algorithm that executes $O(\Delta^3 \log(\frac{n}{\Delta}))$ queries in Subsection~\ref{sec:alg-det}. Both these algorithms require the maximum degree $\Delta$ as part of their inputs. In Subsection~\ref{sec:alg-adaptive}, we show how these algorithms can be turned into adaptive algorithms with a similar number of queries that do not require advanced knowledge of $\Delta$.
\subsection{Randomized Algorithm}\label{sec:alg-rand}
Our algorithm, Algorithm~\ref{alg:main}, executes $\Theta(\Delta^2  \log n)$ queries on random subsets $V_i \subseteq V$, where each vertex is included in $V_i$ with probability $\frac{1}{\Delta+1}$, and outputs every pair $uv \in V \times V$ as an edge of the input graph if $\{u, v \}$ is not contained in any maximal independent set $I_i$, for all $i$, returned by the oracle.

\begin{algorithm}
 \begin{algorithmic}
  \REQUIRE Vertex set $V$, maximum degree $\Delta$, large enough constant $C$
  \FOR{$i = 1 \dots C \cdot (\Delta+1)^2 \cdot \log n$}
    \STATE $V_i \subseteq V$ random sample such that every $v \in V$ is included in $V_i$ with probability $\frac{1}{\Delta+1}$
    \STATE $I_i \gets \text{query}(V_i)$
  \ENDFOR
  \STATE $E \gets \{ uv \in V \times V \ : \ \nexists i \text{ such that } \{u, v\} \subseteq I_i \}$
  \RETURN $(V, E)$
 \end{algorithmic}
 \caption{Randomized graph reconstruction using a \textsf{MIS} oracle \label{alg:main}}
\end{algorithm}

In the analysis, we prove that, for every non-edge $uv \in (V \times V) \setminus E(G)$ in the input graph $G$, both $u$ and $v$ are reported in any independent set $I_i$ with probability $\Theta(\frac{1}{\Delta^2})$. Since $\Theta(\Delta^2 \log n)$ independent sets are computed, any non-edge will be detected with high probability, and by the union bound, this then applies to all non-edges.
\begin{theorem}\label{thm:ub}
 Algorithm~\ref{alg:main} is a non-adaptive randomized query algorithm that executes $O(\Delta^2 \log n)$ \textsf{MIS} queries and correctly reconstructs a graph of maximum degree $\Delta$ with high probability. 
\end{theorem}
\begin{proof}
 Let $G = (V, E)$ denote the input graph, $\Delta$ the maximum degree, and let $\overline{E} = (V \times V) \setminus E$ denote the set of non-edges. We will prove that, with high probability, every non-edge can be identified by the algorithm in that, for a non-edge $uv \in \overline{E}$, there exists an independent set $I_i$ such that $\{u,v\} \in I_i$ with high probability.
 
 Consider thus any non-edge $uv \in \overline{E}$. Then, for every $i$, the probability that both $u$ and $v$ are reported in independent set $I_i$ is at least:
 \begin{align*}
  \Pr[u,v \in I_i] & \ge \Pr[u,v \in V_i \text{ and } (\Gamma(u) \cup \Gamma(v)) \cap V_i = \varnothing] \ , 
  \end{align*}
  since both $u$ and $v$ need to be included in a maximal independent set if they are isolated vertices in $G[V_i]$.
  
  Next, observe that the events ``$u,v \in V_i$'' and ``$(\Gamma(u) \cup \Gamma(v)) \cap V_i = \varnothing$'' are independent since $u$ and $v$ are not adjacent and are thus not contained in $\Gamma(u) \cup \Gamma(v)$. Furthermore, observe that $|\Gamma(u) \cup \Gamma(v)| \le 2 \Delta$. Then:
  \begin{align}
  \Pr[u,v \in V_i \text{ and } (\Gamma(u) \cup \Gamma(v)) \cap V_i = \varnothing] &= \Pr[u,v \in V_i] \cdot \Pr[(\Gamma(u) \cup \Gamma(v)) \cap V_i = \varnothing] \nonumber \\
  &\ge \frac{1}{(\Delta+1)^2} \cdot (1 - \frac{1}{\Delta+1})^{2 \Delta} \nonumber \\
  &\ge  \frac{1}{(\Delta+1)^2} \cdot e^{- \frac{\frac{1}{\Delta+1}}{1- \frac{1}{\Delta+1}}\cdot 2 \Delta} = \frac{1}{(\Delta+1)^2} \cdot \frac{1}{e^2}  \label{eqn:192} \ , 
 \end{align}
 where we used the bound $1-x \ge e^{-\frac{x}{1-x}}$, which holds for every $x < 1$.
 
Next, we compute the probability that both endpoints of the non-edge $uv$ are never reported together:
\begin{align*}  
 \Pr[\{u,v\} \nsubseteq I_i, \text{ for all $i$}] & = \prod_{i=1}^{C (\Delta+1)^2 \cdot \log n} \Pr[\{u,v\} \nsubseteq I_i] \\
 & = \prod_{i=1}^{C (\Delta+1)^2 \cdot \log n} (1 - \Pr[u,v \in I_i]) \\
&  \le (1 - \frac{1}{(\Delta+1)^2} \cdot \frac{1}{e^2})^{C (\Delta+1)^2 \cdot \log n} \\
& \le \exp \left( - \frac{1}{(\Delta+1)^2} \cdot \frac{1}{e^2} \cdot C(\Delta+1)^2 \cdot \log n \right) 
%& = \exp \left( \frac{C}{e^2} \cdot \log n \right) \\
 \le \frac{1}{n^3} \ ,
\end{align*}
where we used the inequality $1+x \le \exp(x)$, and the last inequality holds when $C$ is chosen to be large enough.

%Next, over the course of the $C (\Delta+1)^2 \log n$ iterations executed by the algorithm, we expect both endpoints of the non-edge $uv$ to be reported $C \log(n) / e^2$ times, and by Chernoff bounds, the probability that the endpoints $u,v$ are reported at least once is at least $1-\frac{1}{n^3}$, when $C$ is large enough. 

By the union bound, the probability that the endpoints of at least one of the $O(n^2)$ non-edges are not reported in any independent set $I_i$ is therefore at most $\frac{1}{n}$, which completes the proof.
\end{proof}

\subsection{Deterministic Algorithm}\label{sec:alg-det}
Central to our deterministic algorithm is the notion of a {\em $\Delta$-Query-Scheme} and a {\em Witness}:

\begin{definition}[Witness]
 Let $V$ be a set of $n$ vertices and let $2 \le \Delta \le n/2-1$ be an integer. Then, the tuple $(\{u,v \}, \{w_1, \dots, w_{2\Delta} \})$ with $u,v,w_1, \dots, w_{2\Delta} \in V$ being distinct vertices is called a {\em Witness}, and we denote by $\mathcal{W}$ the set of all witnesses.
\end{definition}

As it will be important in the proof of Lemma~\ref{lem:exist-query-scheme}, we give an upper bound on the number of witnesses:

\begin{lemma}\label{lem:size-W}
 The number of witnesses $|\mathcal{W}|$ is bounded by:
 \begin{align*}
  |\mathcal{W}| \le n^2 \cdot (e \cdot \frac{n-2}{2 \Delta})^{2 \Delta} \ .
 \end{align*}
\end{lemma}
\begin{proof}
 We use the bound ${a \choose b} \le \left(e \cdot \frac{a}{b}  \right)^b$ and obtain:
 \begin{align*}
  |\mathcal{W}| & = {n \choose 2} \cdot {n-2 \choose 2 \Delta} 
  \le n^2 \cdot (e \cdot \frac{n-2}{2\Delta})^{2 \Delta} \ .
 \end{align*}

\end{proof}

\begin{definition}[$\Delta$-Query-Scheme]\label{def:query-scheme}
 Let $V$ be a set of $n$ vertices and let $2 \le \Delta \le n/2-1$ be an integer. The set $\mathcal{Q} = \{Q_1, \dots, Q_{\ell} \}$ is denoted a {\em $\Delta$-Query-Scheme of size $\ell$} if, for every witness $(\{u,v \}, \{w_1, \dots, w_{2\Delta} \}) \in \mathcal{W}$, there exists a query $Q_i \in \mathcal{Q}$ such that:
 
 \begin{enumerate}
  \item $u,v \in Q_i$, and
  \item $\{ w_1, \dots, w_{2\Delta} \} \cap Q_i = \varnothing$.
 \end{enumerate}
 
 In the following, we say that a query $Q_i$ {\em considers} a witness $W \in \mathcal{W}$ if Items~1 and~2 hold for $Q_i$ and $W$.
\end{definition}

We show in Lemma~\ref{lem:query-scheme-alg} that a $\Delta$-Query-Scheme of size $\ell$ immediately yields a non-adaptive deterministic query algorithm for \textsf{GR} for graphs of maximum degree $\Delta$ that executes $\ell$ queries. Our task is thus to design a $\Delta$-Query-Scheme of small size, which we do in the proof of Lemma~\ref{lem:exist-query-scheme}.

\begin{lemma}\label{lem:query-scheme-alg}
 Let $\mathcal{Q}$ be a $\Delta$-Query-Scheme of size $\ell$. Then, there exists a non-adaptive deterministic algorithm for \textsf{GR} that executes $\ell$ queries on graphs of maximum degree $\Delta$.
\end{lemma}
\begin{proof}
 Let $G=(V, E)$ be the input graph, and let $\mathcal{Q}$ be a $\Delta$-Query-Scheme of size $\ell$. The algorithm executes every query $Q_i \in \mathcal{Q}$. Let $I_i$ denote the query answer to $Q_i$. 
 
 We now claim that, for every non-edge $uv \in (V \times V) \setminus E$ in the input graph, there exists an independent set $I_i$ such that $u,v \in I_i$. To see this, denote by $w_1, \dots, w_{\deg(u)}$ the neighbours of $u$ in $G$ and by $w_{\Delta+1}, \dots, w_{\Delta+ \deg(v)}$ the neighbours of $v$ in $G$. Then, since $\mathcal{Q}$ is a $\Delta$-Query-Scheme, there exists a query $Q_i$ such that $u,v \in I_i$, but none of $u$'s and $v$'s neighbours are included. Hence, both $u$ and $v$ are necessarily included in $I_i$ and the algorithm therefore observes a witness that proves that the edge $uv$ does not exist in the input graph. 
 
 Since the argument applies to every non-edge, the algorithm learns all non-edges of the input graph and thus also learns all of the input graph's edges by complementing the set of non-edges.
\end{proof}

\begin{lemma}\label{lem:exist-query-scheme}
 There exists a $\Delta$-Query-Scheme of size $O(\Delta^3 \log \frac{n}{\Delta})$.
\end{lemma}
\begin{proof}
 For an integer $\ell$ whose value we will determine later, let $\mathcal{Q} = \{Q_1, Q_2, \dots, Q_{\ell} \}$ be such that, for every $i$, $Q_i \subseteq V$ is the subset of $V$ obtained by including every vertex with probability $\frac{1}{\Delta+1}$. We use the probabilistic method and prove that $\mathcal{Q}$ is a $\Delta$-Query-Scheme with positive probability, which in turn implies that such a scheme exists.
 
 To this end, let $u, v, w_1 \dots, w_{2\Delta} \in V$ be distinct vertices.  Then, for any $i$, we obtain (the derivation is identical to Inequality~\ref{eqn:192} and therefore not repeated here)
 \begin{align*}
  \Pr[u,v \in Q_i \mbox{ and } w_1, \dots, w_{2\Delta} \notin Q_i] = \frac{1}{(\Delta+1)^2} \cdot (1 - \frac{1}{\Delta+1})^{2 \Delta} \ge \frac{1}{(\Delta+1)^2} \cdot \frac{1}{e^2} \ .
 \end{align*}

Furthermore, the probability that there does not exist a query $Q_i$, for any $i \in [\ell]$, such that $u,v \in Q_i$ and $w_1, \dots, w_{2\Delta} \notin Q_i$ is at most:
\begin{align*}
 \Pr[\nexists i \mbox{ such that } u,v \in Q_i \mbox{ and } w_1, \dots, w_{2\Delta} \notin Q_i]  & \le \left( 1-\frac{1}{(\Delta+1)^2 e^2} \right)^{\ell} \\
 & \le \exp(- \frac{\ell}{(\Delta+1)^2 e^2}) \ .
\end{align*}
Hence, by the union bound over all witnesses $u,v,w_1, \dots, w_{2\Delta}$ (see Lemma~\ref{lem:size-W}), the probability that there exists a witness that is not considered by the $\Delta$-Query-Scheme is at most:
$$n^2 \cdot (e\cdot \frac{n-2}{2 \Delta})^{2 \Delta} \cdot \exp(- \frac{\ell}{(\Delta+1)^2 e^2}) \ . $$
For this probability to be strictly below $1$, it is enough to set
$$\ell  = \Theta( \Delta^3 \cdot \log(\frac{n}{\Delta})) \ , $$
which in turn implies that such a scheme exists.
\end{proof}

Combining Lemma~\ref{lem:query-scheme-alg} and Lemma~\ref{lem:exist-query-scheme}, we obtain the main result of this section.

\begin{corollary} \label{cor:ub-2}
 There exists a deterministic algorithm for \textsf{GR} that executes $O(\Delta^3 \log \frac{n}{\Delta})$ non-adaptive  \textsf{MIS} queries for graphs with maximum degree $\Delta$. 
\end{corollary}
 
 \subsection{Adaptive Algorithms without Knowledge of $\Delta$} \label{sec:alg-adaptive}
 We will now show how the randomized and deterministic algorithms from the previous sections can be turned into adaptive algorithms with similar query complexity that do not require knowledge of $\Delta$ in order to operate. More specifically, let $\mathcal{A}$ be a non-adaptive query algorithm that, given an integer $D$, identifies with high probability in $R(D)$ rounds all non-edges $uv$ such that $\deg(u) \le D$ and $\deg(v) \le D$ hold. It immediately follows from our analyses that both our randomized and deterministic query algorithms have this property when executed with parameter $D$ instead of the true value of $\Delta$. For our randomized algorithm, we have $R(D) = O(D^2 \log n)$, and, for our deterministic algorithm, we have $R(D) = O(D^3 \log(\frac{n}{D})$. We will show how $\mathcal{A}$ can be turned into an adaptive algorithm that does not require knowledge of $\Delta$ and requires overall $O(\mathcal{R}(\Delta))$ rounds.
 
 This is achieved via the doubling strategy displayed in Algorithm~\ref{alg:adaptive-scheme}.

 \begin{algorithm}
 \begin{algorithmic}
  \REQUIRE Vertex set $V$, non-adaptive algorithm $\mathcal{A}$ that identifies all non-edges whose endpoints are of degree at most $D$, for some integer $D$
  \STATE $D \gets 1$, $E' \gets V \times V$
  \WHILE{$\Delta(E') > D$}   
    \STATE $D \gets 2 \cdot D$
    \STATE $F \gets \mathcal{A}(D)$ \COMMENT{Set of non-edges identified}
    \STATE $E' \gets (V \times V) \setminus F$
  \ENDWHILE
  \RETURN $(V, E')$
 \end{algorithmic}
 \caption{Adaptive algorithm that does not require advanced knowledge of $\Delta$\label{alg:adaptive-scheme}}
\end{algorithm}

The algorithm uses the notation $\Delta(E')$, which is to be interpreted as the maximum degree in the graph spanned by the edges $E'$.

 \begin{lemma}
  Let $\mathcal{A}$ be a non-adaptive \textsf{MIS}-query algorithm that, given an integer $D$, in $R(D)$ rounds identifies with high probability all non-edges $uv \in (V \times V) \setminus E$ of the input graph $G=(V, E)$ that have the property that $\deg(u) \le D$ and $\deg(v) \le D$. We assume that $R(D)$ is at least linear in $D$, i.e., $R(D) = \Omega(D)$. Then, there exists an adaptive algorithm for $\textsf{GR}$ that succeeds with high probability, does not require advanced knowledge of $\Delta$, runs in $O(R(\Delta))$ rounds, and requires $O(\log \Delta)$ rounds of adaptivity.
 \end{lemma}

 \begin{proof}
 Let $G=(V, E)$ denote the input graph and $\Delta$ the maximum degree. 
 
 We will argue that when exiting the last iteration of the while loop of the algorithm, the inequalities
 \begin{align}
\Delta \le D < 2 \Delta  \label{eqn:293}
 \end{align}
hold. The lower bound $\Delta \le D$ establishes correctness since the last run of $\mathcal{A}$ is executed with a guess $\Delta \le D$, which implies that all non-edges are correctly identified. The upper bound is required in order to bound the query complexity of the algorithm. 
 
 We first argue the lower bound in Inequality~\ref{eqn:293}. Observe that, in any iteration of the loop, $E' \supseteq E$ holds since $F$ is a subset of the non-edges and $E' = (V \times V) \setminus F$. Thus, we have $\Delta(E') \ge \Delta(E) = \Delta$. When exiting the last iteration of the algorithm, we have $\Delta(E') \le D$ since otherwise this would not be the last iteration of the algorithm. Combining these two inequalities yields $\Delta \le D$. %In the last iteration of the algorithm, The algorithm only stops when $\Delta(E') \le D$, which therefore also implies that $\Delta \le D$. Hence, in the last iteration of the algorithm, we necessarily have 
 
 Regarding the upper bound, we have just proved that, when exiting the last iteration, $\Delta \le D$ holds. It is then also clear that a run with $\Delta \le D < 2 \Delta$ is executed since the guess $D$ is doubled in each iteration. We claim that this run with $\Delta \le D < 2 \Delta$ is indeed the last iteration of the algorithm. Indeed, in this iteration, all non-edges $F$ are correctly identified since $D \ge \Delta$, which implies that $E'$  constitutes the edge set $E$ of the input graph. This further implies that $\Delta(E') = \Delta$. The condition in the while loop then ensures that this is indeed the last iteration of the algorithm.  
 
Last, regarding the runtime of the algorithm, let $D' = 2^i$ be the guess used in the last iteration. We then have that $\Delta \le D' < 2\Delta$. Then, under the assumption that $R(D) = \Omega(D)$, we have: 
\begin{align*}
R(2) + R(4) + \dots + R(D') \le 2 \cdot R(D') \ ,
\end{align*}
which establishes the query complexity of the algorithm. Last, since the while loop is executed $O(\log \Delta)$ times, the algorithm requires only $O(\log \Delta)$ rounds of adaptivity. 
 \end{proof}

 The previous lemma together with our randomized and determinstic algorithms from the previous sections yield the following corollary.
 
\begin{corollary}\label{cor:adaptive-algorithms}
 There are randomized and deterministic adaptive \textsf{MIS}-query algorithms for \textsf{GR} that execute $O(\Delta^2 \cdot \log n)$ and $O(\Delta^3 \cdot \log \frac{n}{\Delta})$ queries, respectively, require $O(\log \Delta)$ rounds of adaptivity, and do not require advanced knowledge of $\Delta$ in order to operate. The randomized algorithm succeeds with high probability. 
\end{corollary}

\section{Lower Bounds}\label{sec:lb}
In this section, we give three lower bound results. First, in Subsection~\ref{sec:lb-det}, we consider the class of non-adaptive deterministic query algorithms and we prove that such algorithms require $\Omega(\Delta^3 / \log^2 \Delta)$ queries. This result renders our deterministic query algorithm optimal, up to poly-logarithmic factors, and it also establishes a separation result between deterministic and randomized query algorithms, since, as demonstrated by our randomized query algorithm, $O(\Delta^2 \log n)$ non-adaptive randomized queries are sufficient.

Next, in Subsection~\ref{sec:lb-rand}, we show that $\Omega(\Delta^2)$ queries are needed for query algorithms that may be adaptive and randomized, and that $\Omega(\log n)$ queries are needed for such algorithms, even if the input graph is an $n$-vertex cycle.
\subsection{Lower Bound for Non-adaptive Deterministic Algorithms} \label{sec:lb-det}
We will first show in Lemma~\ref{lem:lb-query-scheme} that any $\Delta$-Query-Scheme must be of size at least $\Omega(\frac{\Delta^3}{\log^2 \Delta})$. Then, we argue in Lemma~\ref{lem:qs} that the queries executed by any non-adaptive deterministic query algorithm must constitute a $(\Delta-1)$-Query-Scheme. These two lemmas together then imply our main result of this section as stated in Corollary~\ref{cor:det-lb}, i.e., that non-adaptive deterministic query algorithms require $\Omega(\frac{\Delta^3}{\log^2 \Delta})$ queries.

\begin{lemma}\label{lem:lb-query-scheme}
 For every $\Delta \le n^{2/3}$, every $\Delta$-Query-Scheme is of size $\Omega(\frac{\Delta^3}{\log^2(\Delta)})$.
\end{lemma}
\begin{proof}
 Let $C$ be a suitably large constant, and suppose that there exists a $\Delta$-Query-Scheme $\mathcal{Q}$ of size $\ell = \frac{1}{6 \cdot C^2} \cdot \frac{\Delta^3}{\ln^2(\Delta)}$. We will show by contradiction that a $\Delta$-Query-Scheme of this size does not exist.
 Furthermore, we define a relevant query size threshold $t$ by $t := C \cdot \frac{n \ln(\Delta)}{\Delta}$. 
 
 Given $\mathcal{Q}$, we will first argue that there exists a pair of disjoint vertices $x,y \in V$ such that there is no query $Q \in \mathcal{Q}$ with $Q = \{x,y\}$, and there are at most $\Delta / 2$ queries of size at most $t$ that contain both $x,y$. Then, once we have established that such a pair $\{x,y\}$ exists then we use the probabilistic method to show that there exist $2\Delta$ disjoint vertices $w_1, \dots, w_{2\Delta}$ different to $x,y$ such that the witness $(\{x,y\}, \{w_1, \dots, w_{2\Delta} \}$ is not considered in $\mathcal{Q}$, a contradiction to the fact that $\mathcal{Q}$ is a $\Delta$-Query-Scheme. This implies that a query scheme of size $\ell$ cannot exist.
 
 Denote by $\mathcal{X}$ the set of subsets of $V$ of size $2$, i.e., $\mathcal{X} = \{ \{u, v\} \ : \ u,v \in V, u \neq v\}$, where $V$ denotes the set of $n$ vertices of the input graph, and observe that $|\mathcal{X}| = {n \choose 2}$.
 
 First, observe that at most $\ell$ queries in $\mathcal{Q}$ are of size exactly $2$. Hence, as long as $\ell \le {n \choose 2} / 10$, at most a $(1/{10})$-fraction of $\mathcal{X}$ is part of queries of size $2$. Observe that, since the statement of the lemma assumes that $\Delta \le n^{2/3}$, we have 
  \begin{align*}
\ell = \frac{1}{6 \cdot C^2} \cdot \frac{\Delta^3}{\ln^2(\Delta)} \le \frac{\Delta^3}{6 \cdot C^2} \le \frac{n^2}{6 \cdot C^2} \le \frac{n^2}{40} \le \frac{1}{10} {n \choose 2} \ ,    
  \end{align*}
  where we assumed that $C$ is large enough, and the last inequality holds for every $n \ge 2$. Observe that a query $Q = \{u,v\}$ of size $2$ immediately considers all witnesses of the form $(\{u,v\}, \{w_1, \dots, w_{2\Delta} \})$, for any vertices $(w_i)_{1 \le i \le 2 \Delta}$. The previous argument shows that this can happen for at most a $(1/10)$-fraction of pairs $\{u,v \} \in \mathcal{X}$. 
 
 Next, we argue that at least a $(1/4)$-fraction of the pairs $\{u,v\} \in \mathcal{X}$ are such that $\{u, v\}$ is contained in at most $\frac{1}{2} \Delta$ queries of size at most $t$ in $\mathcal{Q}$.
 To this end, observe that any query of size at most $t$ contains at most ${t \choose 2} \le t^2$ distinct pairs $\{u,v\}$, and since there are overall $\ell$ queries, at most 
 $$\ell \cdot t^2 = \frac{1}{6 C^2} \cdot \frac{\Delta^3}{\ln^2(\Delta)} \cdot \left(C \cdot \frac{n \ln(\Delta)}{\Delta} \right)^2 = \frac{1}{6} \cdot n^2 \cdot \Delta$$
 pairs appear overall in all queries of size at most $t$.
 If less than a $(1/4)$-fraction of pairs $\{u,v\}$ were contained in at most $\frac{1}{2} \Delta$ queries of size at most $t$ in $\mathcal{Q}$, then at least a $(3/4)$-fraction of pairs were contained in more than $\frac{1}{2} \Delta$ queries of size at most $t$. But this implies that the following inequality must hold:
 $$(3/4) \cdot {n \choose 2} \frac{1}{2} \Delta \le \ell \cdot t^2 = \frac{1}{6}\cdot  n^2 \cdot \Delta \ . $$
 This inequality, however, does not hold for $n \ge 10$, a contradiction. Hence, we have proved that at least a $(1/4)$-fraction of pairs $\{u,v\} \in \mathcal{X}$ are such that $\{u, v \}$ is contained in at most $\frac{1}{2} \Delta$ queries of size at most $t$ in $\mathcal{Q}$.
 
 Consider thus a pair $\{u,v\} \in \mathcal{X}$ that is included in at most $\Delta/2$ queries of size at most $t$, and that is not included in a query of size $2$. The arguments above ensure that such a pair exists. Denote by $Q_1, \dots, Q_k$ the set of queries that contain both $u$ and $v$, and suppose that $Q_1, \dots, Q_j$ are the queries of size at most $t$ (which implies $j \le \Delta/2$). We now claim that there is a witness $W = (\{u,v\}, \{w_1, \dots, w_{2\Delta} \})$, for some vertices $(w_i)_{1 \le i \le 2\Delta}$, that is not considered by the queries $Q_1, \dots, Q_k$, contradicting the assumption that $\mathcal{Q}$ is a $\Delta$-Query-Scheme, which then completes the proof.
 
 The witness is constructed as follows. For $1 \le i \le j$, let $\tilde{w}_i \in Q_i \setminus \{u, v\}$ be any element. Recall that $\{u,v\}$ is not contained in any query of size exactly $2$, hence, $\tilde{w}_i$ is well-defined. Furthermore, let $\tilde{W} = \{ \tilde{w}_1, \tilde{w}_2, \dots, \tilde{w}_j \}$. We observe that the elements $(\tilde{w}_i)_{1 \le i \le j}$ are not necessarily disjoint, and, hence, $|\tilde{W}| \in  \{1, \dots, j\}$. 
 Next, let $R \subseteq V \setminus \left( \{u,v\} \cup \tilde{W} \right)$ be a random subset of size $2\Delta - |\tilde{W}|$. Then, our witness $W$ is obtained as $W = \{ \{u,v\}, \tilde{W} \cup R \}$, and we will prove in the following that, with positive probability, this witness is indeed not considered by $\mathcal{Q}$. This implies that there exists a witness that is not considered by $\mathcal{Q}$, which completes the proof. 
 
 %Let $W = (\{u,v\}, \{w_1, \dots, w_{2 \Delta} \})$, where, for all $i \le j$, $w_i \in Q_i \setminus \{u,v\}$ is any element, which is possible since $|Q_i| \ge 3$. Hence, by construction, the queries $Q_1, \dots, Q_j$  each contain at least one element of witness $W$, which implies that $W$ is not considered by these queries. The remaining elements $w_{j+1}, \dots, w_{2\Delta}$ are randomly picked. We denote by $R$ the randomly chosen elements $w_{j+1}, \dots, w_{2\Delta}$ in the following. Then, consider any query $Q_i$, with $i > j$, which implies that $|Q_i| \ge t$. Then:
 %$$\Pr[|R \cap Q_i| = \varnothing] \le (1 - \frac{t}{n})^{1.5 \Delta} = (1 - \frac{C \ln \Delta}{\Delta})^{1.5 \Delta} \le e^{-\frac{C \ln \Delta}{\Delta} \cdot 1.5 \Delta} = \frac{1}{e^{1.5 C \ln \Delta}} = \frac{1}{\Delta^{1.5 C}} \ . $$
 
 Consider now any query $Q_i$, for $i \ge j+1$. Then, the probability that $Q_i$ considers $W$ is bounded as follows:
\begin{align*}
\Pr[Q_i & \mbox{ considers } W] \le \Pr[|Q_i \cap R| = \varnothing] \\ 
& = \frac{n - |Q_i| - |\tilde{W}|}{n-2-|\tilde{W}|} \cdot \frac{n - |Q_i| - |\tilde{W}| - 1}{n-2-|\tilde{W}| -1} \cdot \ldots \cdot \frac{n - |Q_i| - |\tilde{W}| - (2\Delta - |\tilde{W}|) + 1}{n-2-|\tilde{W}| - (2\Delta - |\tilde{W}|) + 1} \\
& \le \left( \frac{n - |Q_i| - |\tilde{W}|}{n-2-|\tilde{W}|} \right)^{2\Delta - |\tilde{W}|} \le \left( \frac{n - |Q_i|}{n-2} \right)^{2\Delta - |\tilde{W}|} \le \left( \frac{n - |Q_i|}{n-2} \right)^{1.5\Delta} \\
& = \left( \frac{n - 2 + 2 - |Q_i|}{n-2} \right)^{1.5\Delta} = \left( 1 - \frac{|Q_i|-2}{n-2} \right)^{1.5\Delta} 
\le \left( 1 - \frac{|Q_i|-2}{n} \right)^{1.5\Delta} \\
& \le \exp \left(- \frac{(|Q_i| -2)1.5 \Delta}{n} \right) 
%& = \exp \left(- \frac{(C \cdot \frac{n}{\Delta} \ln(\Delta) -2)1.5 \Delta}{n} \right) \\
 \le \exp \left( - 1.5 C \ln(\Delta) + \frac{3 \Delta}{n}\right) \le \frac{2}{\Delta^{1.5 C}} \ .
%& \le  \frac{{n-t \choose 1.5 \Delta}}{{n \choose 1.5 \Delta}} = \frac{(n-1.5\Delta) \cdot (n-1.5 \Delta-1) \cdot \dots \cdot (n-1.5\Delta-t+1)}{n \cdot (n-1) \cdot \dots \cdot (n-t+1)} \\
% & \le \left( \frac{n-1.5 \Delta}{n} \right)^t = (1 - \frac{1.5 \Delta}{n})^t \le \exp(\frac{1.5 \Delta}{n} \cdot t ) \\
% & = \exp(\frac{1.5 \Delta}{n} \cdot \frac{n C \ln(\Delta)}{\Delta} ) = \exp(1.5 C \ln(\Delta)) = \frac{1}{\Delta^{1.5 C}} \ .
\end{align*}
 
 Since $C$ is a large enough constant, by the union bound, the probability that at least one query of the queries $Q_{j+1} \dots, Q_k$ considers $W$ is therefore strictly below $1$ (recall that there are at most $k \le \ell \le \Delta^3$ queries, and $\Delta^3 \cdot \frac{2}{\Delta^{1.5 C}} < 1$ holds, for large enough $C$). Hence, there exists a witness that is not considered by these queries, which completes the proof.
\end{proof}

\begin{lemma}\label{lem:qs}
 Let $\mathcal{A}$ be a non-adaptive deterministic query algorithm for \textsf{GR} on graphs of maximum degree $\Delta$. Then, the queries executed by $\mathcal{A}$ form a $(\Delta-1)$-Query-Scheme.
\end{lemma}
\begin{proof}
 For the sake of a contradiction, suppose that there exists a sequence of queries $Q_1, \dots, Q_k$ that does not form a $(\Delta-1)$-Query-Scheme but still allows the algorithm $\mathcal{A}$ to learn the input graph exactly. Since the queries do not form a $(\Delta-1)$-Query-Scheme, there exists a witness $W = (\{u,v\}, \{w_1, \dots, w_{2\Delta-2} \})$ that is not considered by the queries. Consider now any two input graphs $G_1$ and $G_2$ of maximum degree $\Delta$ that have the following properties: $u$'s neighbours are $w_1, \dots, w_{\Delta-1}$, and $v$'s neighbours are $w_{\Delta}, \dots, w_{2\Delta-2}$. In $G_1$, there is also an edge between $u$ and $v$, and in $G_2$ there is no edge between $u$ and $v$.  Observe that the degrees of $u$ and $v$ in $G_1$ and $G_2$ are $\Delta$ and $\Delta-1$, respectively. The maximum degrees in $G_1$ and $G_2$ are therefore at most $\Delta$.
 
 Now, we claim that, for every query $Q_i$, the oracle can respond with an independent set $I_i$ that does not include both vertices $u$ and $v$. Observe that the algorithm therefore cannot learn whether the edge $uv$ exists since both $G_1$ and $G_2$ are consistent with all query answers. The algorithm therefore cannot distinguish between the two graphs $G_1, G_2$, which then completes the argument.
 
 Since $Q_i$ does not consider $W$, there exists a vertex $w \in \{w_1, \dots, w_{2\Delta-2} \}$ such that $w \in Q_i$. Hence, the oracle can construct an independent set starting with vertex $w$ (e.g., by running the Greedy maximal independent set algorithm where $w$ is the first vertex picked), which implies that either $u$ or $v$ cannot be included in the independent set. This completes the proof.
 \end{proof}
 \textit{Remark:} The proof of the previous lemma assumes that the oracle can identify a witness not considered by the queries submitted by the algorithm. This is only possible if all queries are submitted simultaneously to the oracle. Observe that this is a valid assumption since we consider the class of non-adaptive algorithms, and such algorithms equally work when all queries are submitted simultaneously.

Combining Lemma~\ref{lem:lb-query-scheme} with Lemma~\ref{lem:qs}, we obtain the main lower bound result of this section as a corollary:
 \begin{corollary}\label{cor:det-lb}
  Every deterministic non-adaptive query algorithm for \textsf{GR} requires $\Omega(\Delta^3 / \log^2(\Delta))$ queries.
 \end{corollary}

\subsection{Lower Bounds for Adaptive Randomized Algorithms} \label{sec:lb-rand}

We first prove that, even on an $n$-vertex cycle, any randomized adaptive algorithm requires $\Omega(\log n)$ queries to solve \textsf{GR}. Our proof is based on an indistinguishability argument: At least $\Omega(\log n)$ queries are needed so that, for each pair of vertices $u,v \in V$, different outcomes from the oracle are observed, which is needed to distinguish all possible cycles from each other.
\begin{theorem}\label{thm:lb-2}
 Every possibly randomized query algorithm with success probability strictly above $1/2$ for \textsf{Graph Reconstruction} using a \textsf{Maximal Independent Set} oracle on an $n$-vertex cycle requires $\Omega(\log n)$ queries. 
\end{theorem}
\begin{proof}
Denote by $\mathcal{C}$ the family of $n$-vertex cycles on the vertex set $V$ with $|V| = n$, which serve as the input to our algorithm. %We consider the uniform input distribution $\mu_0$ on $\mathcal{C}$.

Let $\mathbf{A}$ be a randomized query algorithm that executes $\ell := \log_3(n) - 1$ rounds and that, on any input $C \in \mathcal{C}$, succeeds with probability strictly above $\frac{1}{2}$. Then, by Yao's lemma, there exists a deterministic query algorithm $\mathbf{A}_{\text{det}}$ that succeeds on strictly more than half of the inputs in $\mathcal{C}$ and also runs in $\ell$ rounds. 
  
We first observe that, since $\mathbf{A}_{\text{det}}$ is deterministic, and we also assume that the oracle  answers are deterministic, for every $C \in \mathcal{C}$, there exists a unique execution of $\mathbf{A}_{\text{det}}$, i.e., a sequence of query vertices $V_1, \dots, V_{\ell}$, query answers $I_1, \dots, I_{\ell}$, and output $C_{\text{out}} \in \mathcal{C}$ produced by the algorithm. We observe that $C_{\text{out}}$ may be different from $C$ if the algorithm errs on input $C$.

Denote by $\mathcal{C}_1 \subseteq \mathcal{C}$ the subset of inputs on which the algorithm succeeds, and let $\mathcal{C}_0 = \mathcal{C} \setminus \mathcal{C}_1$ denote the inputs on which the algorithm fails. We will now argue that, for every input $C \in \mathcal{C}_1$ on which the algorithm succeeds, there exists a unique input $C' \in \mathcal{C}_0$ on which the algorithm fails, i.e., there exists an injective mapping from $\mathcal{C}_1$ into $\mathcal{C}_0$, which implies that $|\mathcal{C}_1| \le |\mathcal{C}_0|$. This is a contradiction to the fact that $\mathbf{A}_{\text{det}}$ is correct on strictly more than half of the instances in $\mathcal{C}$, which in turn implies that the algorithm $\mathbf{A}$ does not exist.

Let $C \in \mathcal{C}_1$ be an input on which $\mathbf{A}_{\text{det}}$ succeeds. Denote by $V_1, \dots, V_{\ell}$ the vertices queried by the algorithm in the $\ell$ query rounds and by $I_1, \dots, I_{\ell}$ the query responses. We associate the following complete ternary tree $\mathcal{T}$ with $\ell+1$ layers to the query vertices and responses in an execution of $\mathbf{A}$: 
 \begin{itemize}
  \item The root (layer 1) is labelled with $V = [n]$.
  \item For an internal node in layer $1 \le i \le \ell$ with label $U \subseteq V$, the node has three children with labels $U_1, U_2, U_3$ such that $U = U_1 \ \dot{\cup} \ U_2 \ \dot{\cup} \ U_3$, and 
  \begin{align*}
U_1 & = (U \cap V_i) \cap I_i \ , && \text{queried and reported} \\   
U_2 & = (U \cap V_i) \setminus I_i \ , && \text{queried and not reported}  \\
U_3 & = U \setminus V_i \ . && \text{not queried}  
  \end{align*}
 \end{itemize}
Denote by $L_1, L_2, \dots$ the leaves of $\mathcal{T}$ from left-to-right. We observe that, by construction, the leaves are disjoint and their union equals $V$, i.e., $\dot{\cup}_i L_i = V$. Consider now a leaf $L_i$ that contains at least two vertices $u,v$. The vertices $u,v$ are indistinguishable as they behaved exactly the same throughout the execution of the algorithm. Indistinguishable here means that the cycle $C'$ obtained from the cycle $C$ by swapping the positions of $u$ and $v$ leads to the exact same execution of the algorithm as the execution for $C$. The algorithm, however, fails on $C'$ as the output produced is $C_{\text{out}} = C$. Hence, as long as there exists a leaf that contains at least two vertices, we can identify an input on which the algorithm fails and establish our injective mapping from $\mathcal{C}_1$ to $\mathcal{C}_0$. To avoid that there are no leaves that contain two vertices, the number of leaves in $\mathcal{T}$ must be at least $n$. Since $\mathcal{T}$ is ternary and of depth $\ell+1$, we obtain that $\ell \ge \log_3(n)$ must hold, a contradiction to the assumption that $\ell = \log_3(n) - 1$. This completes the proof.

%Consider now the label $U$ of a leaf, i.e., a node in the last layer $\ell + 1$. Observe that all vertices in $U$ bahaved in exactly the same way throughout the execution. Hence, for any two vertices $u,v \in U$ and any cycle $C \in \mathcal{C}$, the cycle $C$ is indistinguishable from the cycle $C'$ where the positions of $u,v$ are swapped. Hence, if $U$ contains at least $2$ nodes, then the execution necessarily fails on 

%We now claim that the label of every node in the last layer $\ell+1$ (the leaves of $\mathcal{T}$) consists of at most one vertex. For contradiction, suppose that this is not the case and there exists a node in layer $\ell+1$ with a label that contains two distinct vertices $u,v$. Consider the input cycle $P$ and let $P'$ be the cycle obtained from $P$ where the position of vertices $u$ and $v$ are swapped. Then, $\mathbf{A}$ cannot distinguish between $P$ and $P'$ since the two vertices behaved in exactly the same way in all oracle responses, or, in other words, the queries and query responses are equally valid if the input cycle was $P'$ instead of $P$.

%We thus conclude that $\mathcal{T}$ has at least $n$ leaves, which implies that $\mathcal{T}$ has at least $\log_3(n)$ levels, using the fact that $\mathcal{T}$ is ternary. Since $\mathcal{T}$ has $\ell + 1$ layers, we obtain that $\ell \ge \log_3(n)-1$ queries are needed, which completes the proof.
\end{proof}

Last, we give our $\Omega(\Delta^2)$ queries lower bound  for graphs of maximum degree $\Delta$. The key observation in our proof is that, for every non-edge $uv$ in the input graph $G$, there must exist an oracle response maximal independent set that contains both vertices $u$ and $v$, since, if the opposite was true then the algorithm could not distinguish between the input graph $G$ and the graph $G \cup \{uv\}$.

\begin{theorem}\label{thm:lb-1}
 For any $\Delta > 0$, every possibly randomized query algorithm with success probability strictly greater than $1/2$ for \textsf{Graph Reconstruction} using a \textsf{Maximal Independent Set} oracle  requires $\Omega(\Delta^2)$ queries on graphs of maximum degree $\Delta$. 
\end{theorem}
\begin{proof}
 For integers $N > 0$, let $\mathcal{H}_N$ denote the set of all bipartite graphs $H=(A, B, E)$ with $|A| = |B| = N$. Then, let $\mathcal{G}_N$ be the family of graphs obtained from $\mathcal{H}_N$ by turning the bipartitions $A$ and $B$ of each of its graphs $H=(A, B, E) \in \mathcal{H}_N$ into (disjoint) cliques. 
 
 Let $\mathbf{A}$ denote a randomized query algorithm that succeeds with probability strictly above $1/2$ on each input $G \in \mathcal{G}_N$, and for the sake of a contradiction, we assume that $\mathcal{A}$ executes at most $\ell = N^2 - 1$ queries. Then, by Yao's lemma, there exists a deterministic query algorithm $\mathcal{A}_{\text{det}}$ that also executes at most $\ell$ queries and succeeds on strictly more than half of the inputs in $\mathcal{G}_N$. 
 
 First, observe that $\alpha(G) \le 2$, for every $G \in \mathcal{G}_N$ since at most one vertex from $A$ and at most one vertex from $B$ can be included in any independent set. Hence, every independent set reported by $\mathbf{A}_{\text{det}}$ on any of the input graphs $G \in \mathcal{G}_N$ is of size at most $2$.
 
 Denote by $\mathcal{G}_N^1 \subseteq \mathcal{G}_N$ the subset of inputs on which $\mathcal{A}_{\text{det}}$ succeeds, and let $\mathcal{G}_N^0 = \mathcal{G}_N \setminus  \mathcal{G}_N^1$. We will show that there exists an injective map from $\mathcal{G}_N^1$ to $\mathcal{G}_N^0$, which implies that $|\mathcal{G}_N^1| \le |\mathcal{G}_N^0|$, a contradiction to the fact that $\mathcal{A}_{\text{det}}$ succeeds on strictly more than half of the instances. This in turn implies that algorithm $\mathbf{A}$ does not exists, which establishes the theorem.
 
 To this end, let $G \in \mathcal{G}_N^1$ be any instance and consider the execution of $\mathcal{A}_{\text{det}}$ on $G$, i.e., let $V_1, \dots, V_{\ell}$ denote the queries submitted in the $\ell$ query rounds, let $I_1, \dots, I_{\ell}$ denote the query responses, and let $G_{\text{out}} = G$ denote the output produced by the algorithm. 
 
 Let 
 $$F = \{(a,b) \in A \times B \ | \ \{a,b\} \neq I_i, \text{ for all }  i\} \ , $$
 i.e., the set of pairs of vertices $a,b$ that do not constitute a response from the oracle. We observe that every graph $G'$ obtained from $G$ by flipping the edge $ab$, i.e., by introducing $ab$ in $G'$ in case it is not contained in $G$ or by removing it from $G'$ in case it is contained in $G$, leads to the exact same execution of $\mathbf{A}_{\text{det}}$ as all oracle answers are consistent with both $G$ and $G'$. The algorithm however fails on $G'$ since the answer produced by the execution is $G_{\text{out}} = G$. Hence, as long as $F \neq \varnothing$, we can identify a graph $G'$ on which the algorithm fails and map the input $G$ to $G'$ in our injective map. To ensure that $F = \varnothing$, we require at least $|A| \cdot |B| = N^2$ queries. This however is a contradiction to the assumption that only $\ell = N^2 -1$ queries are executed. 
 
 Last, the result follows by observing that the maximum degree $\Delta$ of any graph in $\mathcal{G}_N$ is $2N-1$, which implies that at least $(\frac{\Delta+1}{2})^2$ query rounds are required.
\end{proof}
We remark that Theorem~\ref{thm:lb-1} also shows that $\Omega(m)$ queries are needed for reconstructing graphs on $m$ edges. Observe that every graph of the family of graphs used in Theorem~\ref{thm:lb-1} has $O(\Delta^2)$ edges.

\section{Conclusion}\label{sec:conclusion}
In this paper, we initiated the study of the \textsf{GR} problem using an \textsf{MIS} oracle. We gave a non-adaptive randomized algorithm that reconstructs a graph with maximum degree $\Delta$ using $O(\Delta^2 \log n)$ queries, and a non-adaptive deterministic query algorithm that uses $O(\Delta^3 \log(\frac{n}{\Delta}))$ queries. Both these algorithms require advanced knowledge of $\Delta$ in order to operate, which is unavoidable. We showed that both algorithms can be turned into adaptive algorithms with $O(\log \Delta)$ rounds of adaptivity and a similar number of queries that do not require advanced knowledge of $\Delta$. We also proved that, for adaptive randomized algorithms, $\Omega(\Delta^2)$ queries are necessary, and that such algorithms require $\Omega(\log n)$ queries even if the input graph is an $n$-vertex cycle. Furthermore, we showed that non-adaptive deterministic query algorithms require $\Omega(\Delta^3 / \log^2(\Delta))$ queries, which renders our deterministic algorithm optimal up to poly-log factors.

While, up to lower order terms, the \textsf{MIS} query complexity of \textsf{GR} is settled for randomized algorithms and for non-adaptive deterministic algorithms, it is unclear whether there are adaptive deterministic algorithms that are stronger than non-adaptive deterministic algorithms. More concretely, 
%\begin{enumerate}
 %\item Is there a randomized algorithm that requires only $O(\Delta^2 + \log n)$ queries or can we prove a stronger lower bound than $%\Omega(\Delta^2)$? 
 %\item 
 is there an adaptive deterministic query algorithm that requires fewer than $O(\Delta^3 \log(\frac{n}{\Delta}))$ queries? 
%\end{enumerate}

\bibliography{kot25}

\begin{thebibliography}{10}

\bibitem{ab19}
Hasan Abasi and Nader~H. Bshouty.
\newblock On learning graphs with edge-detecting queries.
\newblock In Aur{\'{e}}lien Garivier and Satyen Kale, editors, {\em Algorithmic
  Learning Theory, {ALT} 2019, 22-24 March 2019, Chicago, Illinois, {USA}},
  volume~98 of {\em Proceedings of Machine Learning Research}, pages 3--30.
  {PMLR}, 2019.

\bibitem{abrs16}
Mikkel Abrahamsen, Greg Bodwin, Eva Rotenberg, and Morten St{\"{o}}ckel.
\newblock Graph reconstruction with a betweenness oracle.
\newblock In Nicolas Ollinger and Heribert Vollmer, editors, {\em 33rd
  Symposium on Theoretical Aspects of Computer Science, {STACS} 2016, February
  17-20, 2016, Orl{\'{e}}ans, France}, volume~47 of {\em LIPIcs}, pages
  5:1--5:14. Schloss Dagstuhl - Leibniz-Zentrum f{\"{u}}r Informatik, 2016.

\bibitem{amm22}
Raghavendra Addanki, Andrew McGregor, and Cameron Musco.
\newblock Non-adaptive edge counting and sampling via bipartite independent set
  queries.
\newblock In Shiri Chechik, Gonzalo Navarro, Eva Rotenberg, and Grzegorz
  Herman, editors, {\em 30th Annual European Symposium on Algorithms, {ESA}
  2022, September 5-9, 2022, Berlin/Potsdam, Germany}, volume 244 of {\em
  LIPIcs}, pages 2:1--2:16. Schloss Dagstuhl - Leibniz-Zentrum f{\"{u}}r
  Informatik, 2022.

\bibitem{acgmw15}
Kook~Jin Ahn, Graham Cormode, Sudipto Guha, Andrew McGregor, and Anthony Wirth.
\newblock Correlation clustering in data streams.
\newblock In Francis~R. Bach and David~M. Blei, editors, {\em Proceedings of
  the 32nd International Conference on Machine Learning, {ICML} 2015, Lille,
  France, 6-11 July 2015}, volume~37 of {\em {JMLR} Workshop and Conference
  Proceedings}, pages 2237--2246. JMLR.org, 2015.

\bibitem{aa05}
Noga Alon and Vera Asodi.
\newblock Learning a hidden subgraph.
\newblock {\em {SIAM} J. Discret. Math.}, 18(4):697--712, 2005.

\bibitem{abkrs04}
Noga Alon, Richard Beigel, Simon Kasif, Steven Rudich, and Benny Sudakov.
\newblock Learning a hidden matching.
\newblock {\em {SIAM} J. Comput.}, 33(2):487--501, 2004.

\bibitem{ac08}
Dana Angluin and Jiang Chen.
\newblock Learning a hidden graph using o(logn) queries per edge.
\newblock {\em J. Comput. Syst. Sci.}, 74(4):546--556, 2008.

\bibitem{akns24}
Sepehr Assadi, Christian Konrad, Kheeran~K. Naidu, and Janani Sundaresan.
\newblock O(log log n) passes is optimal for semi-streaming maximal independent
  set.
\newblock In Bojan Mohar, Igor Shinkar, and Ryan O'Donnell, editors, {\em
  Proceedings of the 56th Annual {ACM} Symposium on Theory of Computing, {STOC}
  2024, Vancouver, BC, Canada, June 24-28, 2024}, pages 847--858. {ACM}, 2024.

\bibitem{as19}
Sepehr Assadi and Shay Solomon.
\newblock When algorithms for maximal independent set and maximal matching run
  in sublinear time.
\newblock In Christel Baier, Ioannis Chatzigiannakis, Paola Flocchini, and
  Stefano Leonardi, editors, {\em 46th International Colloquium on Automata,
  Languages, and Programming, {ICALP} 2019, July 9-12, 2019, Patras, Greece},
  volume 132 of {\em LIPIcs}, pages 17:1--17:17. Schloss Dagstuhl -
  Leibniz-Zentrum f{\"{u}}r Informatik, 2019.

\bibitem{bhrrs18}
Paul Beame, Sariel Har{-}Peled, Sivaramakrishnan~Natarajan Ramamoorthy, Cyrus
  Rashtchian, and Makrand Sinha.
\newblock Edge estimation with independent set oracles.
\newblock In Anna~R. Karlin, editor, {\em 9th Innovations in Theoretical
  Computer Science Conference, {ITCS} 2018, January 11-14, 2018, Cambridge, MA,
  {USA}}, volume~94 of {\em LIPIcs}, pages 38:1--38:21. Schloss Dagstuhl -
  Leibniz-Zentrum f{\"{u}}r Informatik, 2018.

\bibitem{bakaf01}
Richard Beigel, Noga Alon, Simon Kasif, Mehmet~Serkan Apaydin, and Lance
  Fortnow.
\newblock An optimal procedure for gap closing in whole genome shotgun
  sequencing.
\newblock In Thomas Lengauer, editor, {\em Proceedings of the Fifth Annual
  International Conference on Computational Biology, {RECOMB} 2001,
  Montr{\'{e}}al, Qu{\'{e}}bec, Canada, April 22-25, 2001}, pages 22--30.
  {ACM}, 2001.

\bibitem{bk20}
Lidiya~Khalidah binti Khalil and Christian Konrad.
\newblock Constructing large matchings via query access to a maximal matching
  oracle.
\newblock In Nitin Saxena and Sunil Simon, editors, {\em 40th {IARCS} Annual
  Conference on Foundations of Software Technology and Theoretical Computer
  Science, {FSTTCS} 2020, December 14-18, 2020, {BITS} Pilani, {K} {K} Birla
  Goa Campus, Goa, India (Virtual Conference)}, volume 182 of {\em LIPIcs},
  pages 26:1--26:15. Schloss Dagstuhl - Leibniz-Zentrum f{\"{u}}r Informatik,
  2020.

\bibitem{f06}
Uriel Feige.
\newblock On sums of independent random variables with unbounded variance and
  estimating the average degree in a graph.
\newblock {\em SIAM Journal on Computing}, 35(4):964--984, 2006.

\bibitem{g14}
David Galvin.
\newblock Three tutorial lectures on entropy and counting, 2014.

\bibitem{g16}
Mohsen Ghaffari.
\newblock An improved distributed algorithm for maximal independent set.
\newblock In Robert Krauthgamer, editor, {\em Proceedings of the Twenty-Seventh
  Annual {ACM-SIAM} Symposium on Discrete Algorithms, {SODA} 2016, Arlington,
  VA, USA, January 10-12, 2016}, pages 270--277. {SIAM}, 2016.

\bibitem{ggkmr18}
Mohsen Ghaffari, Themis Gouleakis, Christian Konrad, Slobodan Mitrovic, and
  Ronitt Rubinfeld.
\newblock Improved massively parallel computation algorithms for mis, matching,
  and vertex cover.
\newblock In Calvin Newport and Idit Keidar, editors, {\em Proceedings of the
  2018 {ACM} Symposium on Principles of Distributed Computing, {PODC} 2018,
  Egham, United Kingdom, July 23-27, 2018}, pages 129--138. {ACM}, 2018.

\bibitem{gr06}
Oded Goldreich and Dana Ron.
\newblock Approximating average parameters of graphs.
\newblock In Josep D{\'i}az, Klaus Jansen, Jos{\'e} D.~P. Rolim, and Uri Zwick,
  editors, {\em Approximation, Randomization, and Combinatorial Optimization.
  Algorithms and Techniques}, pages 363--374, Berlin, Heidelberg, 2006.
  Springer Berlin Heidelberg.

\bibitem{gk97}
Vladimir Grebinski and Gregory Kucherov.
\newblock Optimal query bounds for reconstructing a hamiltonian cycle in
  complete graphs.
\newblock In {\em Fifth Israel Symposium on Theory of Computing and Systems,
  {ISTCS} 1997, Ramat-Gan, Israel, June 17-19, 1997, Proceedings}, pages
  166--173. {IEEE} Computer Society, 1997.

\bibitem{kmz18}
Sampath Kannan, Claire Mathieu, and Hang Zhou.
\newblock Graph reconstruction and verification.
\newblock {\em {ACM} Trans. Algorithms}, 14(4):40:1--40:30, 2018.

\bibitem{kns23}
Christian Konrad, Kheeran~K. Naidu, and Arun Steward.
\newblock Maximum matching via maximal matching queries.
\newblock In Petra Berenbrink, Patricia Bouyer, Anuj Dawar, and
  Mamadou~Moustapha Kant{\'{e}}, editors, {\em 40th International Symposium on
  Theoretical Aspects of Computer Science, {STACS} 2023, March 7-9, 2023,
  Hamburg, Germany}, volume 254 of {\em LIPIcs}, pages 41:1--41:22. Schloss
  Dagstuhl - Leibniz-Zentrum f{\"{u}}r Informatik, 2023.

\bibitem{mz21}
Claire Mathieu and Hang Zhou.
\newblock A simple algorithm for graph reconstruction.
\newblock In Petra Mutzel, Rasmus Pagh, and Grzegorz Herman, editors, {\em 29th
  Annual European Symposium on Algorithms, {ESA} 2021, September 6-8, 2021,
  Lisbon, Portugal (Virtual Conference)}, volume 204 of {\em LIPIcs}, pages
  68:1--68:18. Schloss Dagstuhl - Leibniz-Zentrum f{\"{u}}r Informatik, 2021.

\bibitem{ms22}
Andrew McGregor and Rik Sengupta.
\newblock Graph reconstruction from random subgraphs.
\newblock In Mikolaj Bojanczyk, Emanuela Merelli, and David~P. Woodruff,
  editors, {\em 49th International Colloquium on Automata, Languages, and
  Programming, {ICALP} 2022, July 4-8, 2022, Paris, France}, volume 229 of {\em
  LIPIcs}, pages 96:1--96:18. Schloss Dagstuhl - Leibniz-Zentrum f{\"{u}}r
  Informatik, 2022.

\bibitem{ms24}
Lukas Michel and Alex Scott.
\newblock Lower bounds for graph reconstruction with maximal independent set
  queries, 2024.

\bibitem{rs07}
Lev Reyzin and Nikhil Srivastava.
\newblock Learning and verifying graphs using queries with a focus on edge
  counting.
\newblock In Marcus Hutter, Rocco~A. Servedio, and Eiji Takimoto, editors, {\em
  Algorithmic Learning Theory}, pages 285--297, Berlin, Heidelberg, 2007.
  Springer Berlin Heidelberg.

\bibitem{rlyw21}
Guozhen Rong, Wenjun Li, Yongjie Yang, and Jianxin Wang.
\newblock Reconstruction and verification of chordal graphs with a distance
  oracle.
\newblock {\em Theoretical Computer Science}, 859:48--56, 2021.

\bibitem{rylw22}
Guozhen Rong, Yongjie Yang, Wenjun Li, and Jianxin Wang.
\newblock A divide-and-conquer approach for reconstruction of {C?5}-free graphs
  via betweenness queries.
\newblock {\em Theoretical Computer Science}, 917:1--11, 2022.

\bibitem{rg20}
V{\'{a}}clav Rozhon and Mohsen Ghaffari.
\newblock Polylogarithmic-time deterministic network decomposition and
  distributed derandomization.
\newblock In Konstantin Makarychev, Yury Makarychev, Madhur Tulsiani, Gautam
  Kamath, and Julia Chuzhoy, editors, {\em Proceedings of the 52nd Annual {ACM}
  {SIGACT} Symposium on Theory of Computing, {STOC} 2020, Chicago, IL, USA,
  June 22-26, 2020}, pages 350--363. {ACM}, 2020.

\end{thebibliography}

\appendix

\section{Independent Set Queries for Graph Reconstruction}\label{app:is-queries}
For completeness, we will now prove that the number of \textsf{Independent Set} queries (or, in fact, any type of query that yields a binary answer) needed for \textsf{GR} on graphs of maximum degree $\Delta \in \Omega(\log n)$ is $\Omega(n \Delta \log(\frac{n}{\Delta}))$.

Angluin and Chen \cite{ac08} observe that, since an independent set query has a binary output, at least $\log |\mathcal{G}|$ queries are needed to distinguish any two graphs in a graph family $\mathcal{G}$. 

We will now prove that the number of $n$-vertex graphs with maximum degree $\Delta$ is $2^{\Omega(n\Delta \log(n/\Delta))}$, which, by the previous argument, implies that $\Omega(n \Delta \log(n/\Delta))$ \textsf{Independent Set} queries are needed to distinguish these graphs. 

Our proof uses entropy-based arguments. We refer the reader to \cite{g14} for an excellent overview of how entropy is connected to counting problems.

\begin{lemma}
 The number of bipartite $2n$-vertex graphs with bipartitions $A$ and $B$ each of size $n$ and with maximum degree $\Delta = \Omega(\log n)$ is at least:
 $$2^{\frac{1}{2} n \Delta \log(\frac{n}{\Delta}) - 2} \ . $$ 
\end{lemma}
\begin{proof}
 We consider the following probabilistic process: Let $G=(A, B, E)$ be a bipartite $2n$-vertex graph with $|A| = |B| = n$ obtained by inserting every potential edge $ab \in A \times B$ into $G$ with probability $\frac{\Delta}{2n}$. Denote by $E$ the indicator random variable of the event that $G$ does not contain a vertex of degree larger than $\Delta$.

 We will now bound the quantity $|\text{range}(G | E = 1)|$ from below, which constitutes a set of bipartite graphs with maximum degree $\Delta$.
 
 To this end, first, observe that:
 \begin{align*}
  \log(|\text{range}(X | E = 1)|) \ge H(X | E = 1) \ ,
 \end{align*}
 which implies that it is enough to bound $H(X | E = 1)$. To bound this quantity, we apply the chain rule for entropy twice on the expression $H(XE)$:
\begin{align*}
 H(XE) &= H(X) + H(E | X) \ , \mbox{ and} \\
 H(XE) &= H(E) + H(X | E) = H(E) + \Pr[E=0]H(X | E=0) + \Pr[E=1] H(X | E=1) \ ,
\end{align*}
which implies:
\begin{align}\label{eqn:392}
 H(X | E=1) & = \frac{H(X) + H(E | X) - H(E) - \Pr[E=0]H(X | E=0)}{\Pr[E=1]} \nonumber \\
 & \ge H(X) + H(E | X) - H(E) - \Pr[E=0]H(X | E=0) \nonumber  \\
 & \ge H(X) - 1 - \Pr[E=0]H(X | E=0) \ ,
 \end{align}
 using the fact that entropy is non-negative, and that the inequality $H(E | X) \le H(E) \le 1$ holds.
 
 Before bounding Inequality~\ref{eqn:392} further, we first prove that $\Pr[E=0]$ is small and we give a bound on $H(X)$.
 
 To see that $\Pr[E=0]$ is small, consider any vertex $v \in A \cup B$. Then, the expected degree of $v$ in $G$ is $\Delta/2$, and, by a Chernoff bound, the probability that the degree of $v$ is larger than $\Delta$ is at most $\frac{1}{2 \cdot n^4}$ (using the assumption that $\Delta = \Omega(\log n)$). By the union bound, the probability that there exists a vertex of degree larger than $\Delta$ is thus at most $\frac{1}{n^3}$, or, equivalently, $\Pr[E=0] \le \frac{1}{n^3}$. 
 
Next, we bound $H(X)$. Since each of the $n^2$ potential edges is included in $G$ independently of all other edges, we obtain:
 \begin{align*}
  H(X) &= n^2 \cdot H_2(\frac{\Delta}{2n}) \\
  & \ge n^2 \log(\frac{2n}{\Delta}) \frac{\Delta}{2n} \ge \frac{1}{2} n \Delta \log(\frac{n}{\Delta}) \ , 
 \end{align*}
 where we bounded the binary entropy function by considering only one of its the two terms. 
 
We are now ready to further simplify Inequality~\ref{eqn:392}, which then yields the result:
 \begin{align*}
 H(X | E=1) & \ge H(X) - 1 - \Pr[E=0]H(X | E=0) \\
 & \ge \frac{1}{2} n \Delta \log(\frac{n}{\Delta}) - 1 -\frac{1}{n^3} \cdot \log(\text{range}(X | E = 0)) \\
 & \ge \frac{1}{2} n \Delta \log(\frac{n}{\Delta}) - 1 -\frac{1}{n^3} \cdot n^2 \\
 & \ge \frac{1}{2} n \Delta \log(\frac{n}{\Delta}) - 1 -\frac{1}{n} \\
 & \ge \frac{1}{2} n \Delta \log(\frac{n}{\Delta}) - 2 \ . 
\end{align*}
 \end{proof}
 
 \begin{corollary}\label{cor:is-queries}
  The number of \textsf{Independent Set} queries needed for \textsf{GR} on $n$-vertex graphs of maximum degree $\Delta = \Omega(\log n)$ is $\Omega(n \Delta \log(\frac{n}{\Delta}))$. 
 \end{corollary}

\end{document}